\documentclass[11pt]{article}
\usepackage{amsmath,amsfonts,amssymb}
\usepackage{amsthm}
\usepackage{fullpage}

\usepackage[numbers,sort]{natbib}

\theoremstyle{plain}
\newtheorem{theorem}{Theorem}[section]
\newtheorem{lemma}[theorem]{Lemma}
\newtheorem{corollary}[theorem]{Corollary}
\newtheorem{proposition}[theorem]{Proposition}

\theoremstyle{definition}
\newtheorem{definition}[theorem]{Definition}
\newtheorem{remark}[theorem]{Remark}

\usepackage{comment}
\usepackage{xspace}
\usepackage{ifthen}
\usepackage{graphicx}
\usepackage{bm}
\usepackage{multirow}
\usepackage{setspace}
\usepackage{verbatim}

\newcommand{\test}[1][]{%
\ifthenelse{\equal{#1}{}}{omitted}{given}%
}

\newenvironment{prevproof}[2]{\noindent {\em {Proof of
{#1}~\ref{#2}:}}}{$\blacksquare$\vskip \belowdisplayskip}

\newcommand{\sign}{\mathrm{sign^+}}

\newcommand{\eat}[1]{}

\newcommand{\R}{{\mathbb{R}}}
\renewcommand{\Pr}{{\textit{\bf Pr}}}
\newcommand{\E}{{\textit{\bf E}}}

\newcommand{\Maj}{{\mathrm{Maj}}}

\newcommand{\prob}[2][]{\text{\bf Pr}\ifthenelse{\not\equal{}{#1}}{_{#1}}{}\!\left[#2\right]}
\newcommand{\expect}[2][]{\text{\bf E}\ifthenelse{\not\equal{}{#1}}{_{#1}}{}\!\left[#2\right]}

\newcommand{\eps}{\epsilon}
\newcommand{\sse}{\subseteq}

\newcommand{\val}{v}
\newcommand{\vals}{{\mathbf \val}}
\newcommand{\valsmi}{{\mathbf \val}_{-i}}
\newcommand{\vali}[1][i]{\val_{#1}}

\newcommand{\vvali}{{\mathbf \val_i}}

\newcommand{\vv}{\varphi}

\newcommand{\vvi}[1][i]{\vv_{#1}}

\newcommand{\Val}{V}
\newcommand{\Vals}{{\mathbf \Val}}

\newcommand{\Vali}[1][i]{\Val_{#1}}

\newcommand{\alloc}{x}
\newcommand{\allocs}{{\mathbf \alloc}}

\newcommand{\alloci}[1][i]{\alloc_{#1}}

\newcommand{\ialloc}{y}
\newcommand{\iallocs}{{\mathbf \ialloc}}

\newcommand{\ialloci}[1][i]{\ialloc_{#1}}

\newcommand{\dens}{f}

\newcommand{\densi}[1][i]{{\dens_{#1}}}

\newcommand{\dist}{F}
\newcommand{\dists}{{\mathbf \dist}}
\newcommand{\distsmi}{{\mathbf \dist}_{-i}}
\newcommand{\disti}[1][i]{{\dist_{#1}}}

\newcommand{\bid}{b}
\newcommand{\bids}{{\mathbf \bid}}

\newcommand{\price}{p}
\newcommand{\prices}{{\mathbf \price}}
\newcommand{\pricei}[1][i]{\price_{#1}}

\newcommand{\iprice}{q}
\newcommand{\iprices}{{\mathbf \iprice}}
\newcommand{\ipricei}[1][i]{\iprice_{#1}}

\newcommand{\vsd}{{\vals \sim \dists}}
\newcommand{\xu}{{x \sim \pmo^n}}
\newcommand{\xo}{{x \sim \zo^n}}

\newcommand{\mech}{(\allocs,\prices)}
\newcommand{\rf}{(\iallocs,\iprices)}

\newcommand{\env}{\mathcal{E}}
\newcommand{\Setting}{\Pi}
\newcommand{\F}{\mathcal{F}}
\newcommand{\C}{\mathcal{C}}
\renewcommand{\L}{\mathcal{L}}
\newcommand{\RR}{\mathbb{R}}
\newcommand{\x}{\mathbf{x}}
\newcommand{\w}{\mathbf{w}}
\newcommand{\OPT}{\mathsf{OPT}}

\newcommand{\gbt}{generalized Border's theorem\xspace}
\newcommand{\gbts}{generalized Border's theorems\xspace}
\newcommand{\lis}{linear inequality system\xspace}
\newcommand{\mem}{{\sc Membership}\xspace}
\newcommand{\optrev}{{\sc OptRev}\xspace}
\newcommand{\optwel}{{\sc OptWel}\xspace}
\newcommand{\stconn}{\ensuremath{\mathsf{\#stConnectivity}\xspace}}

\newcommand{\sharpP}{\ensuremath\mathsf{\#P}}
\newcommand{\Ptime}{\ensuremath\mathsf{P}}
\newcommand{\NP}{\ensuremath\mathsf{NP}}
\newcommand{\coNP}{\ensuremath\mathsf{coNP}}
\newcommand{\PNP}{\ensuremath\mathsf{P^{NP}}}

\newcommand{\Part}{\ensuremath\mathsf{Partition}}
\newcommand{\sPart}{\ensuremath\mathsf{\# Partition}}
\newcommand{\zo}{\{0,1\}}
\newcommand{\pmo}{\ensuremath \{ \pm 1\}}

\title{Public Projects, Boolean Functions,\\ and the Borders of Border's Theorem}
\author{
Parikshit Gopalan\\
Microsoft Research
\and
Noam Nisan\footnote{Partially supported by ISF grants 230/10 and 1435/14 administered by the Israeli Academy of Sciences.}\\
Hebrew University \& \\
Microsoft Research
\and
Tim Roughgarden\footnote{Supported in part by NSF grant and CCF-1215965.}\\
Stanford University
}

\begin{document}
\maketitle

\begin{abstract}
Border's theorem gives an intuitive linear characterization of the
feasible interim allocation rules of a Bayesian single-item
environment, and it has several applications in economic and algorithmic
mechanism design.  All known generalizations of Border's theorem
either restrict attention to relatively simple settings, or resort to
approximation.  This paper identifies a complexity-theoretic barrier
that indicates, assuming standard complexity class separations,
that Border's theorem cannot be extended significantly beyond the
state-of-the-art.  We also identify a surprisingly tight connection
between Myerson's optimal auction theory, when applied to public
project settings, and some fundamental results in the analysis of
Boolean functions.
\end{abstract}

\maketitle

\section{Introduction}

Let us start by considering the computational complexity of a basic problem in
probability theory: characterizing the possible vectors of marginal
probabilities in a probability space.  Questions of this type have already been asked by George Boole in the 19th century (see \cite{P94}) under the name
``conditions of possible experience''.
Here is a simple but very relevant special case that we will focus on:
for which vectors $(p_0, p_1, \ldots, p_n)$
does there exist a probability space with events $E,X_1,\ldots,X_n$, with
$X_1,\ldots,X_n$ independent with $\Pr[X_i]=1/2$ for all $i=1,2,\ldots,n$, 
such that $p_0 = \Pr[E]$ and $p_i= \Pr[E|X_i]$ for all $i=1,2,\ldots,n$? 

The reader may pause for a second here and convince themselves that this
is not a trivial question: for example, $\Pr[E]=1/2$ and
$\Pr[E|X_1]= \Pr[E|X_2]=0.7$ is possible while $\Pr[E]=1/2$ and
$\Pr[E|X_1]= \Pr[E|X_2]=0.8$ is not possible!  


\subsection{Relevance to Bayesian Mechanism Design}


Why would anyone care about this problem?  One motivation comes from
mechanism design, specifically the problem of characterizing the set
of feasible interim allocation rules.
To explain, recall that in a generic Bayesian mechanism design
problem, a principal faces $n$ strategic agents, each holding some
private information, termed its {\em type}, where the tuple of types
is distributed according to some known prior distribution. 
The mechanism must specify an {\em (ex post) allocation rule}: for
each possible 
profile of agent types, an outcome chosen from some family of possible
outcomes or, more generally, a probability distribution over outcomes.
The description of a mechanism is thus naturally exponential in the
number of agents, even if there are only two possible types per agent
and two outcomes.
This exponential description size is the first reason that mechanism
design problems are difficult both mathematically and computationally.
In particular, while it turns out that most mechanism  design problems
of interest are easily expressed as linear programs,
computational efficiency does not follow due to the exponential size
of these linear programs.

There is still hope, however: the goals and constraints of Bayesian
mechanism design problems typically depend only on the marginals of
the allocation rule, also known as the {\em interim allocations}: for
each possible type of each agent, the average outcome over the types
of the other players.   
This reduces the numbers of variables and constraints to be linear in the
number of players rather than exponential.  One would naturally hope
that analysis in terms of the interim allocations helps
the mathematical understanding of the problem,
in addition to computational tractability.  

This is {\em almost} the case.  The only rub is the combinatorial
issue of which interim allocation rules are {\em feasible} ---
for which values of the interim allocation probabilities do there
exist values of the (ex post) allocation probabilities with these
marginals.  Checking feasibility of an interim allocation rule is an
instance of verifying the consistency of a collection of marginals, the
problem described above.

\citet{MR84} were the first to highlight the importance of
characterizing feasible interim allocation rules; their motivation
was to develop an analog of Myerson's characterization of optimal
single-item auctions with risk-neutral bidders~\cite{M81} for the case
of risk-adverse bidders.
\citet{M84} proposed an intuitive necessary condition and conjectured
that it is sufficient for feasibility,
and Border~\cite{B91} proved this conjecture.\footnote{For the finite
  version of Border's theorem that we consider (see
  Section~\ref{ss:border}), there are also more combinatorial
  proofs~\cite{B07,CKM13}.}
For further applications and interpretations in economics,
see~\cite{G13,HR14,M11}.

Border's theorem also has computational implications.  As a linear
characterization that uses only ``simple'' linear inequalities, it
implies that checking the feasibility of an interim allocation rule
is a $\coNP$ problem (assuming finite type-spaces and explicitly given
type distributions). 
In simultaneous and independent works,
\citet{A+12} and \citet{CDW12} show that the problem is in fact in
$\Ptime$.



To what extent can Border's theorem be generalized to other mechanism
design problems?  This question has been the focus of much of the
recent work in algorithmic mechanism design since, as explained above,
it lies at the heart of the efficient computational treatment of  
multi-player mechanism design challenges.  
The current state of knowledge, discussed in detail in
Section~\ref{ss:related}, is that there are analogs of Border's
theorem in settings modestly more general than single-item
auctions, such as $k$-unit auctions with unit-demand
bidders~\cite{A+12,CDW12}, and that approximate versions of Border's 
theorem exist fairly generally~\cite{CDW12,CDW12b}.


Can the state-of-the-art be improved upon?  Can we provide
computationally useful exact extensions of Border's theorem?

\subsection{Summary of Results: The Borders of Border's Theorem}


The first main take-away from this paper is that, under
widely-believed complexity assumptions, Border's theorem cannot be
extended significantly beyond the state-of-the-art.  For example, our
negative results imply the (conditional) impossibility of an exact
Border's theorem even for the following extremely simple mechanism
design setting.

\begin{definition}[Boolean Public Project Problem]
In the {\em Boolean public project} mechanism design problem, there
are only two possible social outcomes: ``yes'' and ``no.''  (E.g.,
build a bridge, or not.)
Each of the $n$ players has valuation $0$ for ``no.'' 
Player $i$'s value for ``yes'' is either 0 or $a_i$; the $a_i$'s are
publicly known, but only player $i$ knows which of 0 or $a_i$ is its
true value for ``yes.''
The distribution on the players' types is
uniform, with all $2^n$ possibilities equally likely.
\end{definition}

Key to our approach is the study of the computational complexity of the
\optrev problem: given a description of 
a mechanism design problem --- for each agent $i$ and type $t_i$, the
probability that $i$'s type is $t_i$ --- compute the maximum expected
revenue obtained by a mechanism that is Bayesian incentive-compatible
and interim individually rational.  (See Section~\ref{s:prelim} for formal
definitions.)  
We prove in Theorem~\ref{t:pp} that the \optrev problem is $\sharpP$-hard\footnote{In this paper all our hardness results are under general Turing reductions.}  
for Boolean public project problems.  A similar result has been independently proven by Yang Cai (private communication).
We note that the \optrev problem is hard despite the setting's status
as ``completely solved'' from a revenue-maximization perspective:
Myerson's optimal auction theory~\cite{M81} tells us exactly what the
optimal auction is (always pick the outcome with the highest sum of ``virtual
values''), and this auction is trivial to implement (since virtual
values are trivial to compute).
Thus while we know what to do (run Myerson's optimal auction) and it is
computationally efficient to do it, it's intractable to compute
(exactly) what our expected revenue will be!\footnote{Solving the
  \optrev problem is non-trivial even when you know what the optimal
  mechanism is, because the expectation is over the exponentially many
  type profiles.}

But so what?  
Isn't 
identifying the optimal mechanism the problem we really care about?  The point is this
(Theorem~\ref{t:main}): 
{\em
a generalization of Border's theorem (defined formally in
Section~\ref{ss:gbt}) would imply that the \optrev
problem is relatively tractable, formally within the complexity class $\PNP$.}
Combining this result with our $\sharpP$-hardness result for the \optrev
problem for Boolean public projects rules out an analog of Border's
theorem for that setting, unless $\sharpP \sse \PNP$, which is widely believed not to be the case.\footnote{Recall that
$\PNP$ denotes the problems that can 
be solved in polynomial time using an $NP$ oracle.  Toda's theorem implies that if $\sharpP \sse
\PNP$, then the polynomial hierarchy collapses to $\PNP$.}

This impossibility result is not an artifact of the fact that 
Boolean public project settings are not ``downward-closed''.  For example,
we can rule out a generalized Border's theorem (here and below, assuming $\sharpP
\not\sse \PNP$) for the setting of single-minded bidders with known
bundles, even when all bundles have size~2 (Theorem~\ref{t:sm}).  The
same reduction rules out a Border's-type theorem for multi-item auctions
with unit-demand bidders (Corollary~\ref{cor:multi}).  Analogous
hardness results even apply to the mathematically well-behaved
class of matroid environments, including
graphical matroids (Theorem~\ref{t:matroid}).

Taken together, our negative results suggest that Border's theorem
cannot be extended significantly beyond the cases already identified
in~\cite{B91,A+12,CDW12,CKM13} without resorting to
approximation (as in~\cite{CDW12,CDW12b}).
In particular, computationally useful Border's-type characterizations
are apparently far rarer than computationally efficient
characterizations of optimal auctions (as in~\cite{M81}).

\subsection{Summary of Results: Applications to Boolean Function Analysis}


The second main take-away of this paper is orthogonal to
the first: there is a surprisingly tight connection between classical
optimal auction theory, especially in the setting of Boolean public
projects, and the analysis of Boolean functions.

To see this, here is yet another formulation of the problem stated at the beginning
of the paper, this time in the
setting of the Fourier transform of Boolean functions.  For 
a Boolean function  $f:\{0,1\}^n \rightarrow \{0,1\}$, we will consider its 
{\em Chow  parameters} (a.k.a.\ level $0$ and $1$ Fourier
coefficients), which are defined as 
\begin{align*}
\hat{f}(0) & = \expect[\xo]{f(x)};\\
\hat{f}(i) & = \expect[\xo]{f(x)(-1)^{1+x_i}} \ \text{for} \ i \in [n].
\end{align*}
We can convexify this set by considering all bounded functions
$f:\zo^n \rightarrow [0,1]$.  The space of Chow parameters
for such functions is a convex polytope which we denote  by $\C_n$.

We observe that the space of feasible interim allocation rules
for the Boolean public project setting (given by the conditional allocations)
is essentially equivalent to the
the space of feasible Chow parameters for functions $f:\zo^n
\rightarrow [0,1]$. This lets us reinterpret our hardness results for
the setting of Boolean Public Projects in the language of Chow parameters as well:
We can view the problem of maximizing revenue as maximizing a
weighted sum of the form $\sum_{i=0}^na_i\hat{f}(i)$ over $\C_n$.  
Theorem \ref{thm:chow} shows that this problem is $\sharpP$-hard,
and hence testing membership of a vector in $\C_n$ is also $\sharpP$-hard.   
At this point we can also return to the problem regarding marginal
probabilities that this paper started with and observe that this also implies its $\sharpP$-hardness.

This equivalence is also useful in the converse direction and provides, in lemma \ref{lem:chow}, 
a simple alternative stand-alone analysis of the \optrev problem
for the Boolean Public Project setting that is based on simple analysis of Boolean functions 
rather than relying on Myerson's analysis of optimal auctions.

\subsection{Related Work}\label{ss:related}

Three recent and independent papers ask to what extent Border's
theorem can be extended beyond the original setting of single-item
auctions~\cite{B91} and provide some positive results.
The results in this paper give senses in which their results
are close to the best possible.

\citet{A+12} give an analog of Border's theorem for every
single-parameter matroid environment, in the form of an
exponential-size set of linear inequalities that characterizes the
feasible interim allocation rules.  
They illustrate some special
cases, such as $k$-unit auctions, in which this characterization can
be used to test the feasibility of a rule, and more generally optimize
over the set of all feasible rules, in polynomial time.  For general
matroids, their linear characterization uses inequalities for which
the right-hand side is $\sharpP$-hard to compute (this follows from the
reductions in the present work); thus for general matroids, their
characterization does not meet our notion of a ``\gbt''
(Definition~\ref{d:gbt}), and indeed by our Theorem~\ref{t:matroid}
it cannot (unless the polynomial hierarchy collapses).

The contributions of \citet{CKM13} relevant to present work are
similar in spirit to but technically different from those of
\citet{A+12}: they give an analog of Border's 
theorem for a class of multi-unit environments, involving
``paramodular'' constraints on the number of units each bidder can
get.\footnote{\citet{CKM13} define paramodularity as the upper bounds
  being submodular, the lower bounds being supermodular, and the two
  constraints being ``compliant,'' meaning irredundant in a certain sense.}
In general, this linear characterization uses inequalities that are
computationally
intractable to compute, but it includes some tractable special cases,
such as when the upper and lower bounds of the different bidders are
uncoupled.

\citet{CDW12} also identify some computationally tractable extensions
of Border's theorem, for example to the case of multi-item auctions
with additive bidders (i.e., no feasibility
constraints).\footnote{Both \citet{A+12} and \citet{CDW12} also give
  (different) compact 
formulations (i.e., polynomially many variables and constraints)
of the feasible interim allocation rules for single-item settings.}
\cite{CDW12,CDW12b} also develop a theory of ``approximate''
Border's-type theorems that encompasses a much wider swath of
mechanism design settings, including all of the concrete settings that
we consider in this paper.
Such approximate theorems identify a set
of linear inequalities with the property that every feasible solution
is close (in $\ell_{\infty}$ norm, say) to a feasible interim
allocation rule, and conversely.
\cite{CDW12,CDW12b} show that approximate versions of Border's
theorem are very useful for algorithmic mechanism design: roughly,
for an arbitrary constant $\eps > 0$, they enable the polynomial-time
computation of a mechanism that is approximately Bayesian incentive
compatible (up to an incentive of $\eps$ to misreport) and is
approximately revenue-optimal (up to a loss of $\eps$).  Our results
in this paper imply that the approximation approach taken in
\cite{CDW12,CDW12b} is unavoidable, in that no exact and
computationally useful analog of Border's theorem can exist for most
of the settings they consider (unless $\sharpP \sse \PNP$).

Our work shares some of the spirit of recent works that use complexity
to identify barriers in mechanism design.
For example, \citet{DDT14} consider a single-bidder multi-item (and
hence multi-parameter) setting,
and prove that it is $\sharpP$-hard to compute a description of the
revenue-maximizing incentive-compatible
mechanism.\footnote{\citet{DDT14} assume an additive bidder, with
the prior distribution over valuations encoded in the natural succinct way.}
Note this problem
is {\em not} hard in most of the (single-parameter) settings that we
consider, where the optimal mechanism is simple to compute and write down
(it is just a ``virtual welfare maximizer'' with virtual values given
by simple explicit formulas, as per
\citet{M81}).  What is hard for us is computing the 
expected revenue obtained by the (simple-to-describe) optimal
mechanisms, not the mechanism design problem per se.  
Still more distant from the present work are previous papers on the
intractability of computing optimal deterministic mechanisms in
various settings,
including~\cite{B08,C+14,DFK11,PP11}.

\section{Preliminaries}\label{s:prelim}


\subsection{Mechanism Design Settings}\label{ss:settings}

We recall first the standard model of binary single-parameter
mechanism design settings,
with players' valuations drawn from a commonly known product
distribution.
Formally, a {\em single-parameter environment} $\env$ consists of the
following 
ingredients: (i) a player set
$U = \{1,2,\ldots,n\}$; (ii) for each player~$i$, a finite set $\Vali$
of possible nonnegative {\em valuations}; (iii) for each~$i \in U$ and
$\vali \in 
\Vali$, a {\em prior} probability $\densi(\vali)$ that $i$'s valuation
is $\vali$; and (iv) a non-empty collection $\F \sse 2^U$ of {\em
  feasible sets.}
For example (see also Section~\ref{s:hard}):
\begin{enumerate}

\item A single-item auction (with $n$ bidders) corresponds to an
  environment with $\F = \{ \emptyset, \{1\}, \{2\}, \ldots, \{n \}
  \}$.

\item In a $k$-unit auction with unit-demand bidders, $\F$ is all
  subsets of $U$ with size at most $k$.

\item In a {\em public project} environment, $\F = \{ \emptyset, U
  \}$.

\item In a {\em single-minded} environment, the set $\F$
  is   described implicitly as follows.  There is a set $M$ of items
  and a
  subset $S_i \sse M$ desired by each bidder $i$.  A set $F \sse U$
  belongs to $\F$ if and only if no desired bundles conflict: $S_i
  \cap S_j = \emptyset$ for all distinct $i,j \in F$.  For example, if
  $|S_i|=2$ for every $i$, then feasible sets correspond to the
  matchings of a graph with vertices $M$ and edges correspondings to
  the $S_i$'s.

\end{enumerate}

We also consider multi-parameter environments.  Because
our primary contributions are impossibility results, we confine
ourselves to the simplest such environments (negative results
obviously apply also to generalizations).
A {\em multi-item auction} environment differs from a
single-parameter environment in the following respects.  First, there
is also a set $M$ of items.  Second, a valuation $\vvali$ of a bidder
$i$ is now a nonnegative vector indexed by $M$.  We restrict attention
to additive bidders, meaning that the value a player $i$ derives
from a set $S \sse M$ of items is 
just the sum $\sum_{j \in S} \val_{ij}$.
Third, $\F$ is now a subset of $2^{U \times M}$, indicating which
allocations of items to bidders are possible.  
For example, the standard setting of unit-demand bidders can be
encoded by defining $\F$ to be the subsets $F$ of $U \times M$ in which,
for every bidder $i$, there is at most one pair of the form $(i,j)$ in
$F$ (and also at most one such pair for each $j$).

By a {\em setting}, we mean a family of mechanism design environments.
For example, all single-item auction environments (with any number~$n$ of
players, any valuation sets, and any prior distribution); all public
project environments; all multi-item unit-demand auction environments;
etc.

\subsection{Bayesian Incentive-Compatible Mechanisms}\label{ss:bic}

This section reviews the classical setup of Bayesian mechanism design
problems, as in Myerson~\cite{M81}.
Fix a binary single-parameter environment, as defined in
Section~\ref{ss:settings}.
A (direct-revelation) {\em mechanism} $\mech$
comprises an {\em allocation rule} $\allocs:\Vals \rightarrow
\{0,1\}^n$ and a {\em payment 
  rule} $\prices:\Vals \rightarrow \RR^n_+$,
where $\Vals = \Val_1 \times \cdots \times \Val_n$.
The former is a map (possibly randomized) from each bid vector~$\bids$
--- with one bid per player --- to a characteristic vector of a feasible
set in $\F$, the latter 
is a map (possibly randomized) from each bid vector $\bids$ to a 
payment vector $\prices$, with one payment per player.  
For the questions we study, we can restrict attention to truthful
mechanisms (via the Revelation Principle, e.g.~\cite{N07}),
and we henceforth use the true valuations $\vals$ in place of the
bids $\bids$.

A mechanism $\mech$ and prior distribution $\dists$ over valuations
together induce an {\em interim allocation rule}
\begin{equation}\label{eq:iar}
\ialloci(\vali) = \expect[\valsmi \sim
  \distsmi]{\alloci(\vali,\valsmi)},
\end{equation}
which describes the probability (over the randomness in $\valsmi$ and
any randomness 
in $\allocs$) that bidder~$i$ is chosen when it reports the valuation
$\vali$.  Similarly, the {\em interim payment rule}
\begin{equation}\label{eq:ipr}
\ipricei(\vali) = \expect[\valsmi \sim
  \distsmi]{\pricei(\vali,\valsmi)}
\end{equation}
describes the expected payment made by~$i$ when it reports $\vali$.
The pair $\rf$ is the {\em reduced form} of the mechanism $\mech$.
We sometimes call $\allocs$ and $\prices$ {\em ex post} rules 
for emphasis.

A mechanism $\mech$ for an environment is {\em Bayesian incentive compatible
  (BIC)} if truthful bidding is a Bayes-Nash equilibrium.
We assume that bidders are risk-neutral and have quasi-linear utility
(value minus payment), and can therefore use linearity of expectation
to write succinctly the BIC condition in terms of the reduced form
$\rf$ of a mechanism:
\begin{equation}\label{eq:bic}
\vali \ialloci(\vali) - \ipricei(\vali)
\ge
\vali \ialloci(\vali') - \ipricei(\vali')
\end{equation}
for every bidder $i$, true valuation $\vali$, and reported valuation
$\vali'$.

Similarly, we can express the {\em interim individual rationality
  (IIR)} requirement --- stating that truthful bidding 
leads to non-negative expected utility --- by
\begin{equation}\label{eq:iir}
\vali \ialloci(\vali) - \ipricei(\vali)
\ge
0
\end{equation}
for every bidder $i$ and true valuation $\vali$.

We can also write the seller's expected revenue
\begin{equation}\label{eq:rev}
\sum_{i=1}^n \sum_{\vali \in \Vali} \densi(\vali) \cdot \ipricei(\vali)
\end{equation}
in terms of the reduced form of the mechanism.
In the {\em optimal (revenue-maximizing) mechanism design   problem},
the goal is to identify the BIC and IIR mechanism $\mech$ that
maximizes~\eqref{eq:rev}.

Multi-item auction environments can be treated similarly, with
an (ex post) allocation rule choosing a feasible set of $\F$ for each
valuation profile $\vals$ and
$\alloc_{ij}(\vals)$ now denoting whether or not bidder~$i$ receives
the item~$j$.  (There is no need to keep track of a
separate payment 
for each item received.)  For example, because we assume that bidders
are additive, the BIC constraints for a
multi-item environment can phrased in terms of the reduced form $\rf$ of
a mechanism by
\begin{equation}\label{eq:bic2}
\sum_{j \in M} \val_{ij} \ialloc_{ij}(\vvali) - \ipricei(\vvali)
\ge
\sum_{j \in M} \val_{ij} \ialloc_{ij}(\vvali') - \ipricei(\vvali')
\end{equation}
for every bidder $i$, true valuation $\vvali$, and reported valuation
$\vvali'$.

\subsection{Border's Theorem for Single-Item Auction
  Environments}\label{ss:border}

As discussed in the Introduction, 
there are several applications that rely on understanding
when 
an interim allocation rule $\iallocs$ is induced by some ex post
allocation rule~$\allocs$.  Such interim rules are said to be {\em
  feasible}.

{\em Border's Theorem}~\cite{B91} 
characterizes interim feasibility for single-item
auction environments.  To review it, fix such an environment and
assume without loss of generality 
that the valuation sets $V_1,\ldots,V_n$ are disjoint.\footnote{We can
  enforce this by thinking of each $\vali \in \Vali$
as the pair $\{ \vali,i\}$.}
To derive an obvious necessary condition for feasibility, consider an
ex post allocation rule $\allocs$ with induced interim rule
$\iallocs$.
Fix for each bidder $i$ a set $S_i \sse V_i$ of {\em distinguished}
valuations. 
Linearity of expectation implies that the probability, 
over the random valuation profile
$\vals \sim \dists$ and any coin flips of the allocation rule
$\allocs$, that the winner of the item has a distinguished valuation
is
\begin{equation}\label{eq:lhs}
\sum_{i=1}^n \sum_{\vali \in S_i} \densi(\vali) \ialloc_{i}(\vali).
\end{equation}
The probability, over $\vsd$, that
there is a bidder with a distinguished type is
\begin{equation}\label{eq:rhs}
1 - \prod_{i=1}^n \left( 1 - \sum_{\vali \in S_i} \densi(\vali) \right).
\end{equation}
Since there can only be a winner with a distinguished type when there is
a bidder with a distinguished type, the quantity in~\eqref{eq:lhs} can
only be less than~\eqref{eq:rhs}.  Border's theorem asserts that 
satisfying these
linear (in $\iallocs$) constraints, ranging over all choices of $S_1
\sse V_1,\ldots,S_n \sse 
V_n$, is also a sufficient condition for the feasibility of an interim
allocation rule~$\iallocs$.

\begin{theorem}[Border's Theorem~\cite{B91}]\label{t:border}
In a single-item environment,
an interim allocation rule $\iallocs$ is feasible if and only if for
every choice $S_1 \sse V_1,\ldots,S_n \sse V_n$ of distinguished
types,
\begin{equation}\label{eq:border}
\sum_{i=1}^n \sum_{\vali \in S_i} \densi(\vali) \ialloc_{i}(\vali)
\le
1 - \prod_{i=1}^n \left( 1 - \sum_{\vali \in S_i} \densi(\vali) \right).
\end{equation}
\end{theorem}


\section{Generalized Border's Theorems and Computational Complexity}


\subsection{Generalizing Border's Theorem}\label{ss:gbt}


What do we actually mean by a ``Border's-type theorem?''  
Since we aim to prove impossibility results, we should adopt a
definition that is as permissive as possible.
Border's theorem (Theorem~\ref{t:border}) gives a characterization of
the feasible interim allocation rules of a single-item environment as
the solutions to a finite system of linear inequalities.  This by
itself is not interesting --- since the set is a polytope,\footnote{The
  set of ex post allocation rules is a polytope, and the feasible
  rules are the image of this polytope under a linear map (and hence
  also a polytope).} it is guaranteed to have such a characterization
(see e.g.~\cite{ziegler}).  The appeal of Border's theorem is that the
characterization uses only the ``nice''
linear inequalities in~\eqref{eq:border}.
Our ``niceness'' requirement is that the characterization use only
linear inequalities that can be efficiently recognized and tested.
This is a weak necessary condition for such a characterization to be
computationally useful.

\begin{definition}\label{d:gbt}
A {\em
  \gbt} holds for the mechanism design setting~$\Setting$ if, for every
environment $\env \in 
\Setting$, 
there is a \lis $\L(\env)$ 
such that the following properties hold.
\begin{enumerate}

\item (Characterization) For every $\env \in \Setting$, the feasible
  solutions of $\L(\env)$ 
  are precisely the feasible interim allocation rules of $\env$.

\item (Efficient recognition) There is a polynomial-time algorithm
  that, given as input a 
  description of an environment $\env \in \Setting$ and a
linear inequality, decides whether or not it belongs to
$\L(\env)$.\footnote{Note this is a weaker assumption than requiring the
  efficient recognition of an arbitrary valid inequality.}
Note the description length of $\env$ is polynomial in $n$, the
$|\Vali|$'s, and the maximum number of bits needed to describe a 
valuation or a prior probability.  

\item (Efficient testing)
There is a polynomial $p(\cdot)$ such that,
for every $\env \in \Setting$, the 
natural encoding length of every 
inequality of $\L(\env)$ is at most $p(\ell)$, where $\ell$
is the description length of $\env$.  
(The number $|\L(\env)|$ of inequalities can still be exponential.)
Thus, deciding whether or not a given point $\x \in \RR^n$ satisfies a
given inequality of $\L(\env)$ can be done in time polynomial in the
descriptions of $\env$ and $\x$, just by computing and comparing the
two sides of the inequality.


\end{enumerate}
\end{definition}

For example, consider the original Border's theorem, for single-item 
auction environments (Theorem~\ref{t:border}).
The recognition problem is straightforward: the left-side
of~\eqref{eq:border} encodes the $S_i$'s, from which the right-hand
side can be computed and checked in time polynomial in the description
of $\env$.  It is also evident that the inequalities
in~\eqref{eq:border} have a polynomial-length description.

The characterization in Theorem~\ref{t:border} and the extensions
in~\cite{A+12,CDW12,CKM13} have additional features
not required or implied by Definition~\ref{d:gbt}, such as
polynomial-time separation oracles (and even compact
reformulations in the single-item
case~\cite{A+12,CDW12}).\footnote{Separation can be hard even 
  when recognition and testing are easy --- see e.g.~\cite{GLS} for
  some examples in combinatorial optimization.}  
All of our impossibility
results rule out analogs of Border's theorem that merely satisfy
Definition~\ref{d:gbt}, let alone these stronger properties.

A \gbt does imply that the problem of
testing the feasibility of an interim allocation rule is 
in $\coNP$.  
To prove that such a rule for an environment
$\env$ is not feasible, one simply exhibits an inequality of
$\L(\env)$ that the rule fails to satisfy --- there is always such an
inequality by Definition~\ref{d:gbt}(i), and verifying this failure
reduces to the recognition and testing problems for $\L(\env)$, which
by Definition~\ref{d:gbt}(ii,iii) are polynomial-time solvable.

Formally, we define the \mem problem for a setting $\Setting$ as:
given as input a description 
of an environment $\env$ and an interim allocation rule $\iallocs$ for
it, decide whether or not $\iallocs$ is feasible.
\begin{proposition}\label{prop:conp}
If a \gbt holds for the mechanism design setting~$\Setting$, then
the \mem problem for $\Setting$ belongs to $\coNP$.
\end{proposition}

\subsection{Impossibility Results from Computational Intractability}

We now forge a connection between the existence of \gbts and the
computational complexity of natural optimization problems.
For a setting $\Setting$, the \optrev($\Setting$) problem is: given a
description of an environment $\env \in \Setting$, compute the
expected revenue earned by the optimal BIC and IIR mechanism.
The main result of this section shows that a \gbt exists for a setting
only when it is relatively tractable to solve exactly
the \optrev problem.

\begin{theorem}\label{t:main}
If a mechanism design setting $\Setting$ admits a \gbt,
then the \optrev($\Setting$) problem belongs to
$\PNP$.\footnote{Recall that $\PNP$ denotes the problems that can
  be solved in polynomial time using an $\NP$ oracle (or equivalently,
  a $\coNP$ oracle).}
\end{theorem}

We later apply Theorem~\ref{t:main} in the form of
the following corollary.

\begin{corollary}\label{cor:main}
If the \optrev($\Setting$) problem is $\sharpP$-hard,
then there is no \gbt for $\Setting$ (unless the
polynomial hierarchy collapses).
\end{corollary}

In the next section, we prove that the \optrev problem is $\sharpP$-hard
for many simple settings, ruling out the possibility of \gbts for them
(conditioned on $\sharpP \not\sse \PNP$).

\begin{remark}
By the same reasoning and under the same complexity assumption,
$\sharpP$-hardness of the \optrev($\Setting$) problem rules out any $\mathsf{PH}$
algorithm that recognizes the set of interim allocation rules for the
setting $\Setting$, not just via a \gbt. Modulo the same
assumptions, it also rules out other approaches to efficient revenue
optimization, say via an extended formulation of polynomial size.  

\end{remark}

\vspace{.1in}
\noindent
\begin{prevproof}{Theorem}{t:main}
Consider a setting $\Setting$ in which a \gbt holds and an instance
of the \optrev($\Setting$) problem --- a description of an environment
$\env \in \Setting$.  We compute the optimal expected revenue of a BIC
and IIR mechanism via linear programming, as follows.

The decision variables of our linear program correspond to the
components of an interim allocation rule $\iallocs$ and payment rule
$\iprices$.  The number of variables is polynomial in $n$ and $\max_i
|\Vali|$ and hence in the description of $\env$.  The (linear)
objective function is to maximize the expected seller revenue, as
in~\eqref{eq:rev}.  The BIC and IIR constraints can be expressed as
a polynomial number of linear inequalities as in~\eqref{eq:bic}
and~\eqref{eq:iir}, respectively.
By assumption, the 
interim feasibility constraint
can be expressed by a \lis 
$\L(\env)$ that satisfies the properties of Definition~\ref{d:gbt}.
Thus the optimal objective
function value of the linear program~(LP) that maximizes~\eqref{eq:rev}
subject to~\eqref{eq:bic}, \eqref{eq:iir}, and $\L(\env)$
is the solution to the given instance of \optrev($\Setting$).

To solve~(LP), we turn to the ellipsoid method~\cite{K79}, which
reduces the solution of a linear program to a polynomial number of
instances of a simpler problem (plus polynomial additional
computation).  The relevant simpler problem for us 
is a {\em membership oracle:} given an alleged reduced form $\rf$,
decide whether or not $\rf$ is feasible for~(LP).  Using
Proposition~\ref{prop:conp} and the fact that there are only
polynomially many constraints of the form~\eqref{eq:bic}
and~\eqref{eq:iir}, we have a $\coNP$ membership oracle.
(The most common way to apply the ellipsoid method, on the other hand, 
involves a {\em separation oracle}: given an alleged reduced form $\rf$,
either verify that $\rf$ is feasible for~(LP) or, if not,
produce a constraint of~(LP) that it violates.  Our assumption of a 
\gbt for $\Setting$ does not include also a separation oracle of
the same complexity hence we use the version of Ellipsoid that is based on a membership oracle.)

For optimization over polytopes described by linear inequalities of
bounded size, assuming that one knows a priori a 
feasible point $(\iallocs_0,\iprices_0)$,
the ellipsoid method can also be used to 
reduce the solution of a linear program to a polynomial number of
invocations of a membership oracle (see~\cite[P.189]{S86}).
The size bound on the defining linear inequalities is implied by
the Efficient Testing condition in Definition \ref{d:gbt}.
Computing a feasible point is trivial in our context: we can just consider
a mechanism that outputs some constant outcome irrespectively of players' types
(with payments that are always zero) and the induced constant interim allocation rule. 

We conclude that the linear program~(LP) and hence the \optrev problem
can be solved using a
polynomial number of
invocations of a membership oracle and polynomial
additional computation.  Since the membership problem for~(LP) belongs
to $\coNP$, the \optrev problem belongs to $\PNP$.
\end{prevproof}


What we have actually shown is a general Turing reduction from the 
\optrev($\Setting$) problem to the \mem problem for $\Setting$.

\begin{corollary}\label{cor:mem-rev}
If the \optrev($\Setting$) problem is $\sharpP$-hard,
then so is the \mem problem for $\Setting$.
\end{corollary}

More generally, the proof of Theorem~\ref{t:main} shows that a \gbt
allows an arbitrary linear function of the interim allocation and
payment rules to be optimized over the space of BIC and IIR mechanisms
in $\PNP$.  
For example, let \optwel($\Setting$) be the problem of, given an
environment $\env \in \Setting$, computing the maximum expected
welfare achieved by a BIC and IIR mechanism.\footnote{The
  welfare-maximizing mechanism is of course the VCG mechanism
  (e.g.~\cite{N07}) --- but even knowing this, it is not generally
  trivial to compute its expected welfare.}
Since the expected welfare obtained by a mechanism can be written as
$\sum_{i=1}^n \sum_{\vali \in \Vali} \dens_i(\vali)\vali\ialloc_i(\vali)$
for a single-parameter environment or as
$\sum_{i=1}^n \sum_{\vvali \in \Vali} \dens_i(\vvali)
\sum_{j \in M} \val_{ij}\ialloc_{ij}(\vvali)$ for a multi-item
environment, we have the following corollary.

\begin{corollary}\label{cor:main2}
If the \optwel($\Setting$) problem is $\sharpP$-hard,
then there is no \gbt for $\Setting$ (unless the
polynomial hierarchy collapses).
\end{corollary}

\section{Complexity of Computing the Optimal Expected
  Revenue and Welfare}\label{s:hard}

This section shows that,
in several simple settings, the \optrev or the \optwel problem
is $\sharpP$-hard.  
Together with Corollaries~\ref{cor:main} and~\ref{cor:main2}, these
results effectively rule out, conditioned on $\sharpP
\not\sse \PNP$, significant generalizations
of Border's theorem beyond those that are already known.



\subsection{Preliminaries and Examples}


Recall the definition of the \optrev problem for a setting $\Setting$:
given a description of an environment $\env \in \Setting$, compute the
maximum expected revenue obtained by any BIC and IIR mechanism.  Even
if one is handed the optimal mechanism on a silver platter, and this
mechanism runs in polynomial time for every valuation profile,
naive computation of its expected revenue requires running over the
exponentially many valuation profiles.  
The question is whether or not there are more efficient methods for
computing this expected value.  

To develop context and intuition for the problem,
we review the argument that
the \optrev problem for single-item auctions can be
solved in polynomial time
(see also~\cite{A+12,CDW12}).\footnote{This is not surprising in light of
Theorems~\ref{t:border} and~\ref{t:main}!}
\begin{enumerate}

\item For each bidder~$i$ and possible valuation $\vali \in \Vali$,
compute the corresponding {\em (ironed) virtual valuation} $\vvi(\vali)$, as in
Myerson~\cite{M81}.\footnote{There is an analogous simple formula for
  the case of a discrete set of bidder valuations~\cite{E07}.}  This
can be done straightforwardly in time polynomial in the size of
$\env$.

\item By~\cite{E07,M81}, the (ex post) allocation rule $\allocs^*$
of the optimal mechanism awards the
  good to the bidder with the highest positive virtual valuation
  (breaking ties lexicographically, say), if any, and to no one
  otherwise.  The second step of the algorithm is to compute the
  interim allocation rule $\iallocs^*$ induced by $\allocs^*$.  This
  can be done in polynomial time via a simple
  computation.\footnote{For a bidder $i$ and valuation $\vali$,
    $\ialloci^*(\vali)$ is 0 if $\vvi(\vali) \le 0$, and otherwise is
    $\prod_{j \neq i} \prob{\vv_j(\val_j) < \vvi(\vali)}$.}

\item By~\cite{E07,M81}, the solution to the \optrev problem equals
  $\sum_{i=1}^n \expect[\vali \sim
  \disti]{\vvi(\vali)\ialloci^*(\vali)}$; given the virtual valuations
  and the interim allocation rule, this quantity is trivial to
  compute in polynomial time.

\end{enumerate}




\subsection{Public Projects}\label{ss:pp}


Recall that in a public project environment, there are only two
outcomes: choose all players (``build the bridge'') or no player
(``not'').  We now show that the \optrev problem is hard in such
environments, even in the extremely simple case when each
player is equally likely to have a zero valuation or a known positive
valuation for the ``yes'' outcome.\footnote{Note that the \optwel problem
  is trivial to solve in the public project setting.}

\begin{theorem}\label{t:pp}
The \optrev problem is $\sharpP$-hard for the public project setting, even
when every player has only two possible valuations, and the valuation
distribution is uniform.
\end{theorem}
This result can be usefully re-interpreted in the
context of Boolean function analysis; see Section~\ref{s:bool} for
details.

Theorems~\ref{t:main} and~\ref{t:pp} immediately imply the following.
\begin{corollary}\label{cor:pp}
Unless the polynomial hierarchy collapses, there is no \gbt for the
public project setting.
\end{corollary}

We begin by reformulating the \optrev problem, in the special case of an
environment $\env$ in which each bidder~$i$ is equally likely to have
the valuation~0 or the valuation~$a_i$.  We show it is equivalent to
the $\sharpP$-complete problem of computing the Khintchine constant of a
vector, which is defined as follows (cf., \cite{DDS}). 

\begin{definition}
For a vector $a \in \R^n$, define 
\[ K(a) = \expect[\xu]{|x \cdot a|} \]
to be the {\em Khintchine constant} for $a$.
\end{definition}

It is known that 
\[ \frac{\|a\|_2}{\sqrt{2}} \leq K(a) \leq \|a\|_2,\]
where the upper bound follows from Cauchy-Schwarz
and the lower bound is the
classical Khintchine inequality.  We use the following
$\sharpP$-hardness result, which we prove in the appendix for
completeness. 

\begin{lemma}\label{l:wb}
Given a vector $a \in \R^n$, the problem of computing $K(a)$ is
$\sharpP$-complete, even when $a \in
\mathbb{Z}^n$ with bit-length polynomial in $n$.
\end{lemma}

\begin{lemma}
\label{lem:key}
The optimal revenue of a BIC and IIR mechanism for the public projects problem is $K(a)/2$.
\end{lemma}

Combining these two lemmas, the proof of Theorem \ref{t:pp} is
immediate. We present two proofs of Lemma \ref{lem:key}. Our first
proof invokes Myerson's characterization of optimal
auctions \cite{M81}. Our second proof is self-contained and  uses an
argument from the analysis of Boolean functions \cite{O'Donnell} and will be presented in section \ref{s:bool}.  

\vspace{.1in}

\begin{prevproof}{Lemma}{lem:key}
(First version.)
The standard virtual valuations for our setting
(see~\cite{M81,E07}) are $\vv_i(0) = -a_i$ and 
$\vv_i(a_i) = a_i$ for each bidder~$i$.  
In a binary single-parameter environment, the revenue-maximizing
auction always selects the outcome that maximizes the sum of the
virtual valuations of the chosen players --- the ``virtual welfare''
--- and the expected revenue of this auction equals its expected
virtual welfare~\cite{M81,E07}.  Translated to the current special
case, the solution to the \optrev problem is precisely
\begin{equation}\label{eq:pp}
\expect[\vsd]{\max \left\{ 0, \sum_{i=1}^n \vv_i(\vali) \right \}}
=
\expect[\xu]{\max \left\{ 0, x \cdot a\right\} },
\end{equation}
where $x \in \{ \pm 1 \}^n$ is chosen uniformly at random and $a =(a_1,\ldots,a_n)$.

Since $(- x) \cdot a = -(x \cdot a)$, we have
\[ \expect[\xu]{\max \left\{ 0, x \cdot a\right\} } =
\frac{1}{2}\expect[\xu]{\left|x \cdot a\right|}  = \frac{K(a)}{2},\]
completing the proof.
\end{prevproof}

\subsection{Single-Minded Environments}\label{ss:sm}

This section presents an impossibility result for a downward-closed
setting (where, unlike in public project environments, every subset of
a feasible set is again feasible).  Recall that in a single-minded
environment, each bidder $i$  wants a publicly known subset $S_i$ of goods, and the feasible
outcomes are subsets of bidders with mutually disjoint bundles.

Assume that there are $t$ items in total, and $p$ players, so
that we can view each player's bundle as an edge in a graph with $t$
vertices. Denote the resulting graph by $H$. We allow parallel
edges (i.e., players that desire the same bundle). The
prior distribution is uniform over $\zo^p$. Given a
string $x \in \zo^p$, we define a subgraph $H_x$ of $H$ by keeping edge $i$ if
$x_i =1$. Having $x_i =0$  indicates that player $i$ has $0$
valuation, $x_i =1$ means her valuation is $1$ for her bundle. A
feasible allocation corresponds to a matching in $H_x$, and the
maximum welfare is the size of a maximum matching. Thus the \optwel
problem in this setting amounts to computing the expected size of the
maximum matching in a random edge subgraph of $H$.

\begin{theorem}\label{t:sm}
The \optwel problem is $\sharpP$-hard for the single-minded bidder
setting, even when every player's bundle contains two items,
every player has only two possible valuations, and the valuation
distribution is uniform.
\end{theorem}

\begin{proof}
The proof is by reduction from the $s$-$t$ connectivity problem in
directed graphs.  Formally,
\stconn~is the following problem: Given a directed graph $H$ with two
distinguished vertices 
$s,t$, what is the probability that there is an $s$-$t$ path in a random
edge subgraph of $H$?
Valiant shows this problem is $\sharpP$-complete \cite{valiant}.

Assume that $G$ has vertex set $[n]\cup \{s, t\}$ 
and $m$ edges. We construct a bipartite
graph $H$ where the vertex set is $L \cup R$, where $L = \{s\} \cup
[n]$ and $R = [n] \cup \{t\}$.  
When we speak of an edge $(i,j)$ in $H$, we always mean with
$i \in L$ and $j \in R$. 
If there is a directed edge $(i,j)$
in $G$, we add an edge between $i$ on the left and $j$ on the right in
$H$. We call these red edges.
For every $i \in \{1,\ldots,n\}$, we add $k$ parallel edges between the
two copies of $i$
($k$ will be a large polynomial in $m, n$). We call these blue edges.

Consider picking a random subgraph $H'$ of $H$.
Except with probability $m/2^k$ (over the choice of blue edges), at
least one blue edge of the form $(i,i)$ survives for every $i \in [n]$. 
Conditioned on this event, the maximum matching
in $H'$ has size at least $n$.
Further, a matching of size $n+1$ exists if and only if the subgraph
of $G$ that corresponds to the surviving red edges has an $s$-$t$ path $P$
--- the red edges corresponding to $P$ match $s \in L$, $t \in R$,
and both copies of all vertices internal to $P$ (vertices not on $P$
are matched with blue edges).

Let $p$ denote the probability that a random subgraph $G'$ of $G$
contains an $s$-$t$ path. Let the random variable $M$ denote the size of the maximum
matching in $H'$. The observations above imply that
\begin{align*}
\expect[H']{M} & \leq p(n+1) + (1- p)n = n +p\\
\expect[H']{M} & \geq (p - \frac{m}{2^k})(n+1) + (1- p)n = n + p -
\frac{mn}{2^k}.
\end{align*}
Note that $p$ is always an integer multiple of $1/2^m$. Therefore if
we choose $k$ sufficiently large (bigger than $mn$), then we can recover $p$
by subtracting $n$ from $\expect[H']{M}$ and rounding up the
remainder to the form $c/2^m$ for some integer $c$.
\end{proof}

Combining Corollary~\ref{cor:main2} and Theorem~\ref{t:sm} gives the
following.
\begin{corollary}\label{cor:sm}
Unless the polynomial hierarchy collapses, there is no \gbt for the
single-minded bidder setting.
\end{corollary}


\subsection{Multi-Item Auctions with Unit-Demand Bidders}

Recall that in a multi-item auctions setting, there is a set $M$ of goods and
each bidder $i$ has a valuation $v_{ij}$ for each item $j \in M$.  The
set $\F \sse 2^{U \times M}$ describes the feasible allocations of
items to bidders.  In a multi-item auction with unit-demand bidders,
the feasible sets $\F$ correspond to the matchings of a complete
bipartite graph $G$ with vertex sets $U$ and $M$.
Given valuations $\vals$,
the welfare-maximizing allocation corresponds to the maximum-weight
matching in $G$ (with each edge weight $w_{ij}$ equal to the valuation
$\val_{ij}$).  Since the reduction in the proof of Theorem~\ref{t:sm}
produces a bipartite graph, the reduction also implies that the
\optwel problem is $\sharpP$-hard for the setting of multi-item auctions
with unit-demand bidders.  Corollary~\ref{cor:main2} then implies the
following.

\begin{corollary}\label{cor:multi}
Unless the polynomial hierarchy collapses, there is no \gbt for the
setting of multi-item auctions with unit-demand bidders.
\end{corollary}

\subsection{Matroid Environments}

Our final example shows that there are even simple matroid settings
which do not admit \gbts (unless $\sharpP \sse \PNP$).  In a {\em
  graphical matroid} environment, bidders correspond to the edges of
an undirected graph $G=(V,E)$. 
The feasible sets correspond to the
acyclic subgraphs of $G$, so the welfare-maximizing outcome
corresponds to a maximum-weight spanning forest.
Using the same valuation distributions as in Section~\ref{ss:sm},
the \optwel problem becomes that of computing the expected size of a
spanning forest --- or equivalently, 
the expected number of connected components --- of a random subgraph
$G'$ of $G$.





\begin{theorem}\label{t:matroid}
The \optwel problem is $\sharpP$-hard for the graphical matroid
setting, even when every player has only two possible valuations, and
the valuation distribution is uniform.
\end{theorem}

\begin{proof}
We reduce from the \stconn\ problem in undirected graphs, which is also
$\sharpP$-complete~\cite{valiant}.
Given an instance $G$ of \stconn, 
let $C_1$ denote the expected number of connected components in a
random subgraph of~$G$.
Derive $H$ from $G$ by adding the edge $(s,t)$ --- a second copy if it
is already in $G$ --- and 
let $C_2$ denote the expected number of connected components in a
random subgraph of~$H$.
Since
\[ C_1 - C_2 = \frac{1-p}{2},\]
where $p$ is the probability that $s$ and $t$ are connected in a
random subgraph of $G$, we conclude that the \stconn\ problem reduces
to the \optwel problem in the graphical matroid setting.
\end{proof}


\begin{corollary}
Unless the polynomial hierarchy collapses, there is no \gbt for the
graphical matroid setting.
\end{corollary}

\section{Connections to Boolean Function Analysis}\label{s:bool}

In this section we re-interpret our results on the public projects
problem in terms of Boolean function analysis and, conversely, provide a stand-alone analysis of the 
optimal-revenue Boolean public project mechanism based on Boolean function analysis.  

\subsection{From Boolean Public Projects to Boolean Functions}

Let us go back to the revenue maximization problem, $\optrev$, for the Boolean public project 
problem addressed in subsection \ref{ss:pp} and re-derive the characterization of
the optimal revenue from first principles.  Recall that each of our $n$ players has either value $0$ or 
value $a_i$ for the ``positive outcome'' with both possibilities equally 
likely and independent of the others' values.  To convert this setting to a Boolean functions setting 
let us indicate by $x_i=0$ the case that the value of player $i$ is 0 and
by $x_i=1$ the case that the value of player $i$ is $a_i$.  Our prior distribution of 
values for the Boolean public project setting is now translated to the uniform distribution
over the Boolean hypercube $\{0,1\}^n$.  

Let us denote by $f(x_1,\ldots,x_n) \in [0,1]$ the probability of a positive outcome that our 
mechanism gives when the players' values are according to indicators $x_i$.
Let us further denote by 
\[ f^i(x_i) = \E_{x_{-i} \in \{0,1\}^{n-1}} [f(x_i,x_{-i})] = \E_{x \in \{0,1\}^n} [f(x) | x_i].\]
the interim allocation of player $i$ with value indicated
by $x_i$. Now let us denote by $p^i(x_i)$ the interim payment of player $i$ with value $x_i$.  
Our incentive constraints and individual rationality constraints imply
bounds on $p^i$ in terms of the $f^i$ and allow us to characterize exactly the maximum 
payments that are possible for a given interim allocation rule.

\begin{lemma}
The maximum possible interim payments for a Bayesian Incentive 
Compatible and Interim Individually Rational Mechanism for the Boolean Public Projects problem 
with interim allocations given by $f^i$ with $f^i(1) \ge f^i(0)$ are
precisely $p^i(0)=0$ and $p^i(1)=a_i \cdot (f^i(1)-f^i(0))$.  In particular, 
the optimal revenue among such mechanisms is exactly
$\sum_i (a_i/2) \cdot (f^i(1)-f^i(0))$.
\end{lemma}

\begin{proof}
The individual rationality constraint for $x_i=0$ immediately implies $p^i(0) \le 0$ 
since in that case the player gets no utility from ``the bridge''.  Now let us focus on
the incentive constraint for the case $x_i=1$: reporting the truth will result in 
the bridge being built with probability $f^i(1)$ while lying and reporting 0 will
give a probability $f^i(0)$.  Player $i$'s value for telling the truth is thus 
$a_i \cdot (f^i(1)-f^i(0))$ larger than his value from lying, and this difference is the maximum
that the payment $p^i(1)$ can be larger than the payment $p^i(0)$ 
without his utility becoming lower which would violate the incentive constraints.

We now only need to observe that indeed setting $p^i(0)=0$ and $p^i(1)=a_i \cdot (f^i(1)-f^i(0))$ 
does yield an incentive compatible and individually rational
mechanism.  Individual rationality: for $x_0=0$ the value and payment and thus also 
utility are 0; for $x_i=1$ the value is $a_i f^i(1)$ and the payment $p^i(1)$ is lower,
and thus the utility is positive.  Incentives: for $x_i=0$ lying would not increase player 
$i$'s value which is 0, but may only increase his payment; for $x_i=1$ the
payment $p^i(1)$ was chosen exactly so his utility for the truth will exactly match 
his utility from lying, $a_i \cdot f^i(0)$.
\end{proof}

The constraint that $f^i(1) \ge f^i(0)$ for each $i$ must be satisfied by 
every incentive compatible mechanism since otherwise the incentive constraints for 
$x_i=0$ would dictate $p^i(1) \ge p^i(0)$ while the incentive 
constraints for $x_i=1$ would dictate $p^i(1) < p^i(0)$, a contradiction.  We may however,
without loss of generality, optimize over all functions $f$, since one 
can always switch the roles of $0$ and $1$ by "Not"-ing an input bit.  We thus have reduced the
revenue maximization problem to the following problem on functions $f: \{0,1\}^n \rightarrow [0,1]$.

\begin{lemma}\label{l:opt-bool}
The optimal revenue of the Boolean public project problem
is given by 
\begin{align}
\label{eq:opt1}
 \OPT = \max_{f:\{0,1\}^n \rightarrow [0,1]} \sum_{i=1}^n \frac{a_i}{2}(\expect[\xo]{f(x)|x_i =1}- \expect[\xo]{f(x)|x_i = 0}).
\end{align}
\end{lemma}

It turns out that characterizing the function $f$ that maximizes this weighted sum of 
differences is not difficult using the basic tools of Boolean function 
analysis (and of course yields the same mechanism as
Myerson's theorem implies).  The rest of this section will continue coming up with this derivation in a leisurely way within the 
Boolean function context.  As expected, this analysis will identify the optimal $f$ to be threshold function (halfspace)
$f(x)=\sign(\sum a_i \cdot (-1)^{x_i+1})$ where $\sign(z)=1$ if $z \geq 0$ and $\sign(z)=0$ otherwise.  
The fact that this is a Boolean function (rather than taking fractional values)
translates to the optimal mechanism being deterministic, and the 
fact that it is a monotone function can be translated to an optimal mechanism that
is truthful in dominant strategies (obtained, as usual, by setting critical payments). 

\subsection{The Chow Parameters of a Boolean Function}

We refer to functions $f:\zo^n \rightarrow [0,1]$ as bounded
functions and the subset of functions $f:\zo^n \rightarrow \zo$ as Boolean functions.
Every bounded function can be viewed as a convex combination of
Boolean functions. Given $f:\zo^n \rightarrow [0,1]$, we define its {\em Chow
  parameters} or degree-$0$ and degree-$1$ Fourier coefficients as\footnote{The reason we use $(-1)^{1 + x_i}$ instead of
      the more common $(-1)^{x_i}$ is that $\hat{f}(i)$ being positive
      implies positive correlation between $x_i$ and $f$ by Equation \eqref{eq:fi}.}
\begin{align*}
\hat{f}(0) & = \expect[\xo]{f(x)};\\
\hat{f}(i) & = \expect[\xo]{f(x)(-1)^{1+x_i}} \ \text{for} \ i \in
    [n].
\end{align*}
We refer to $(\hat{f}(0),\ldots,\hat{f}(n))$ as the {\em Chow vector} of $f$.
Let us define the set
\[ \C_n = \{(c_0,\ldots,c_n)| \exists f:\zo^n \rightarrow [0,1] 
\ s.t. \ c_i = \hat{f}(i) \ \text{for} \ 0 \leq i\leq n \}. \]
Note that the space of feasible Chow vectors is a polytope, since
it is convex, and has finitely many vertices corresponding to the Chow
vectors of Boolean functions. Let us denote this polytope by
$\C_n$. While this polytope is natural in the context of
Fourier analysis, we are not aware of prior work that studies it
explicitly.


Note that 
\begin{align}
\hat{f}(0) & = \expect[\xo]{f(x)}\notag\\
& = \frac{1}{2}\expect[\xo]{f(x)|x_i =1} + \frac{1}{2}\expect[\xo]{f(x)|x_i =0}.\\
\hat{f}(i) & = \expect[\xo]{f(x)(-1)^{1 +x_i}} \notag\\
& = \frac{1}{2}\expect[\xo]{f(x)|x_i =1}
- \frac{1}{2}\expect[\xo]{f(x)|x_i =0}. \label{eq:fi}
\end{align}


This lets us reinterpret our results regarding revenue maximization for the public projects problem in the
language of Chow parameters. By comparing Equation \eqref{eq:opt1} 
in the statement of lemma \ref{l:opt-bool} and \eqref{eq:fi}, we see that
the problem of maximizing revenue is equivalent to maximizing the
weighted sum of the Chow parameters of a function. 

\begin{align}
\label{eq:chow}
 \OPT = \max_{f:\{0,1\}^n \rightarrow [0,1]} \sum_{i=1}^n a_i \hat{f}(i).
\end{align}

We briefly explain why these parameters are of interest in the analysis
of Boolean functions. Recall that $\sign(z)=1$ if $z \geq 0$ and $\sign(z)=0$ otherwise. We say that
a function $h:\zo^n \rightarrow \zo$ is a halfspace if there exist
real numbers $a_0,\ldots,a_n$ such that
\[ h(x) = \sign(a_0 + \sum_{i=1}^n a_i(-1)^{1+ x_i}). \]
We may assume that $a_0 + \sum_{i=1}^n a_i(-1)^{1+ x_i}$ never
vanishes on $\zo^n$.
An elegant result of Chow \cite{Chow61} implies that the
the Chow parameters of a halfspace identify it uniquely in the set of
all Boolean functions, and in fact all Bounded functions. 

\begin{theorem}[Chow's Theorem ~\cite{Chow61}]
Let $h:\zo^n \rightarrow \zo$ be a halfspace. If $f:\zo^n
\rightarrow [0,1]$ satisfies $\hat{f}(i) = \hat{h}(i)$ for $0 \leq i
  \leq n$, then $f(x) = h(x)$ for all $x \in \zo^n$.
\end{theorem}

Chow's theorem is usually stated assuming that $f$ is Boolean, but
the proof \cite[Theorem 5.1]{O'Donnell} also applies to the bounded
case. Chow's argument does not give an algorithm to reconstruct a
halfspace from its Chow parameters; this problem is known as the Chow parameters problem. It was solved 
recently by \citet{OS08} and subsequently improved in
\citet{DDFS}. Both results start from approximations to the Chow
parameters, and return a halfspace that is close in Hamming distance
to the target halfspace (with exactly the right Chow parameters). 

Further motivation comes from the fact that for
monotone functions, $2\hat{f}(i)$ equals the influence of variable
$x_i$, and hence $2\sum_i\hat{f}(i)$ equals the average sensitivity of
the function $f$. For $n$ odd, let $\Maj:\zo^n \rightarrow \zo$ denote the Majority
function. It is known \cite[Theorem 2.33]{O'Donnell} that for  all Boolean functions, 
\begin{align}
\label{eq:monotone}
\sum_{i=1}^n \hat{f}(i) \leq \sum_{i=1}^n \hat{\Maj}(i) =
\left( \sqrt{\frac{2}{\pi}}+o(1)\right)\sqrt{n}
\end{align}
which implies that the Majority function has the highest average sensitivity among
all monotone functions\footnote{Going back the the Boolean public projects problem, this would be an estimate of the
optimal revenue for the case where all $a_i=1$.} . For more on Chow parameters and their significance, we refer the
reader to \cite[Chapter 5]{O'Donnell}. 


\subsection{Optimization over $\C_n$}

At this point we proceed with our analysis in the setting of Chow parameters which, on one hand, applies our hardness results 
to the problems of membership an optimization over the polytope $\C_n$ (Theorem \ref{thm:chow}), and on the other hand completes the promised stand-alone alternative
proof of 
Lemma \ref{lem:key}. 
By now, using lemma \ref{l:opt-bool} and the formulation in equation \ref{eq:chow} 
we have shown the equivalence of the revenue optimization problem for the 
Boolean public project setting to that of 
maximizing the weighted sum of the Chow parameters of a 
bounded function, equivalently maximizing a linear function over the polytope $\C_n$.

Given $a  = (a_0,\ldots,a_n)\in \R^{n+1}$, it defines a linear
objective function over $\C_n$ given by $a\cdot c$. 
To analyze the linear function corresponding to a vector $a \in
\R^{n+1}$, we  define the affine function $a(x) = a_0 +
\sum_{i=1}^na_i(-1)^{1+x_i}$ mapping $\zo^n$ to $\R$.  
When $a_0 =0$, the first part of Lemma \ref{lem:chow} below is
essentially a restatement of Lemma \ref{lem:key}, whereas the second part
replaces the role of Myersons's theorem for identifying the optimal function in \ref{lem:key}.

\begin{lemma}
\label{lem:chow}
Given $a \in \R^{n+1}$, we have
\begin{align}
\label{eq:chow2} 
\max_{c \in \C_n} a \cdot c = \expect[\xo]{\sign(a(x))a(x)} = \frac{1}{2}(K(a) + a_0).
\end{align}
Equality is attained at the Chow vectors corresponding to functions
$f:\zo^n \rightarrow [0,1]$ such that $f(x) = \sign(a(x))$ whenever $a(x) \neq 0$. 
\end{lemma}
\begin{proof}
Let $c =(\hat{f}(0),\ldots,\hat{f}(n))$. We have
\begin{align*}
a\cdot c & = a_0\hat{f}(0) + \sum_{i=1}^na_i\hat{f}(i)\notag\\
& = a_0\expect[\xo]{f(x)} + \sum_{i=1}^na_i \expect[\xo]{f(x)(-1)^{1+ x_i}}\notag\\
& = \expect[\xo]{f(x)(a_0 + \sum_{i=1}^na_ix_i)}\notag\\
& = \expect[\xo]{f(x)a(x)}.
\end{align*}
Since $0 \leq f(x) \leq 1$, 
\[ f(x)a(x) \leq \sign(a(x))a(x).\]
If $a(x) \neq 0$, this holds with equality iff $f(x) =
\sign(a(x))$, whereas if $a(x) = 0$, $f(x)$ can take an arbitrary
value in $[0,1]$. This gives us 
\[ \max_{c \in \C_n} a \cdot c = \expect[\xo]{\sign(a(x))a(x)}\]
and characterizes the functions that achieve equality. To complete the
proof, we just compute this expectation. This is a routine calculation
which we defer to Appendix \ref{app:exp}.
\end{proof}

\begin{theorem}
\label{thm:chow}
The problems of linear optimization over $\C_n$ and deciding membership in $\C_n$ are $\sharpP$-hard. 
\end{theorem}
\begin{proof} 
The hardness of linear optimization follows from Lemmas \ref{lem:chow}
and \ref{l:wb}: if we can solve linear optimization efficiently, we
can compute $K(a)$.

The hardness of membership is proved using a similar argument to
Theorem \ref{t:main}. 
$0^{n+1}$ is a feasible point in $\C_n$.
Hence the ellipsoid method can also be used to 
reduce linear optimization to a polynomial number of
invocations of a membership oracle (see~\cite[P.189]{S86}).
Hence if we can decide membership in the polytope $\C_n$,
then we can solve linear optimization over $\C_n$ using the Ellipsoid
algorithm, which we just showed is $\sharpP$-hard. 
\end{proof}

This hardness result rules out a {\em nice} characterization
of the polytope $\C_n$, in the spirit of Definition \ref{d:gbt}. We
believe this negative result is interesting in the context of the Chow
parameters problem, and sheds light on why the exact version of the
problem, (where the goal is to find a function whose Chow vector
equals the input) is hard.

\citet{OS08} observe that the {inverse problem} of
computing $\hat{f}(0)$ for a given halfspace is $\sharpP$-complete, which 
implies that given a target Chow vector, it is hard to verify if an
input halfspace has exactly these Chow parameters. This can be viewed
as evidence that the exact version of the Chow parameters problem is intractable.

Assume that we drop the requirement that the output be a halfspace and are willing to settle for any bounded
Boolean function that has some compact representation, which 
lets us evaluate its Chow vector exactly. Can we now hope to solve the Chow
problem exactly? Theorem \ref{thm:chow} says this is unlikely,
since such an algorithm with allow us to test membership in $\C_n$:
run the algorithm, and check the function output by it.

Let us return to the question that we have started the paper with:
For which vectors $(p_0, p_1, ..., p_n)$ does there exist a probability space with events $E,X_1,\ldots,X_n$,
with $X_1,\ldots,X_n$ independent and $\Pr[X_i]=1/2$ for all $i= \in [n]$,
such that $p_0 = \Pr[E]$ and $p_i= \Pr[E|X_i]$ for all $i \in [n]$?
Let $x_i$ be the indicator of the event $X_i$. Define $f:\zo^n
\rightarrow [0,1]$ by setting $f(x)$ to be the
probability of $E$ given $x$. The above problem reduces to testing
whether a vector of Chow parameters is feasible, which is $\sharpP$-hard
by Theorem \ref{thm:chow}.

\subsection{The vertices of $\C_n$ are halfspaces}

Lastly, our results yield a characterization of the
vertices of the polytope $\C_n$. Recall that a point $c \in \C_n$ is a
vertex iff there exists $a \in \R^{n+1}$ such that 
\begin{align}
\label{eq:a} 
a\cdot c > a\cdot c' \  \ \forall c' \neq c \in \C_n.
\end{align}

\begin{theorem}
\label{thm:polytope}
The vertices of $\C_n$ are in 1$-$1 correspondence with halfspaces.
\end{theorem}
\begin{proof}
Consider the halfspace 
\[ h(x) = \sign(a(x))\]
where $a(x) = a_0 + \sum_{i=1}^na_i(-1)^{1+x_i}$ does not vanish at any point in $\zo^n$. Lemma \ref{lem:chow}
shows that the linear function specified by the vector $a =
(a_0,\ldots,a_n)$ is maximized over $\C_n$ uniquely at the Chow vector
of $h(x)$, which is therefore a vertex of $\C_n$.

For the other direction, fix a vertex $c$ of $\C_n$. There exists a linear function specified
by $a \in \R^{n+1}$ such that  $a\cdot c > a\cdot c'$ for  $c' \neq c \in \C_n$.
By Lemma \ref{lem:chow}, a bounded function $f$ with Chow vector
$c$ must satisfy $f(x) = \sign(a(x))$ for all $x$ where $a(x) \neq 0$. 
Thus $f$ has been specified except for points where $a(x)$ vanishes.
We  claim that the choice of $a$ implies that $a(x) \neq 0$ for all 
$x \in \zo^n$. Assume for contradiction that $a$ vanishes at some
subset of $\zo^n$. 
Define Boolean functions $f_1$ and $f_0$ as
\begin{align*}
f_1(x) = \begin{cases}  1 & \ \text{if} \ a(x) =0\\
\sign(a(x)) & \ \text{otherwise}
\end{cases}\\
f_0(x) = \begin{cases}  0 & \ \text{if} \ a(x) =0\\
\sign(a(x)) & \ \text{otherwise}
\end{cases}
\end{align*}

We claim that $\hat{f_1}(0) \neq \hat{f_0}(0)$, thus their Chow
vectors are distinct. This holds because\footnote{Alternately, we can observe that both
  $f_1$ and $f_0$ are halfspaces, and Chow's theorem tells us that
  their Chow vectors are distinct.} 
\[ \expect[\xo]{f_1(x)} - \expect[\xo]{f_0(x)} =  \Pr_{x \in \zo^n}[a(x) =0] > 0. \]
By construction, both $\hat{f_1}$ and $\hat{f_0}$ maximize the objective
function $a$.  This contradicts Equation \eqref{eq:a}.
We conclude that $a(x) \neq 0$ for all $x \in \zo^n$, hence 
$f(x) = \sign(a(x))$ is the unique bounded function with Chow vector $c$.
\end{proof}


\section*{Acknowledgments}

We would like to thank  Noga Alon, Yang Cai,  Costis Daskalakis, Uri
Feige,  Gil Kalai, Nati Linial, Yuval Peres,  Rocco Servedio, Li-Yang
Tan, and, Matt Weinberg for helpful discussions during various stages
of this work.

\bibliographystyle{chicago}
\bibliography{border}

\begin{thebibliography}{}

\bibitem[\protect\citeauthoryear{Alaei, Fu, Haghpanah, Hartline, and
  Malekian}{Alaei et~al.}{2012}]{A+12}
Alaei, S., H.~Fu, N.~Haghpanah, J.~D. Hartline, and A.~Malekian (2012).
\newblock Bayesian optimal auctions via multi- to single-agent reduction.
\newblock In {\em {ACM} Conference on Electronic Commerce, {EC} '12}, pp.\ ~17.

\bibitem[\protect\citeauthoryear{Border}{Border}{1991}]{B91}
Border, K.~C. (1991).
\newblock Implementation of reduced form auctions: A geometric approach.
\newblock {\em Econometrica\/}~{\em 59\/}(4), 1175--1187.

\bibitem[\protect\citeauthoryear{Border}{Border}{2007}]{B07}
Border, K.~C. (2007).
\newblock Reduced form auctions revisited.
\newblock {\em Economic Theory\/}~{\em 31}, 167--181.

\bibitem[\protect\citeauthoryear{Briest}{Briest}{2008}]{B08}
Briest, P. (2008).
\newblock Uniform budgets and the envy-free pricing problem.
\newblock In {\em Automata, Languages and Programming, 35th International
  Colloquium, {ICALP}}, pp.\  808--819.

\bibitem[\protect\citeauthoryear{Cai, Daskalakis, and Weinberg}{Cai
  et~al.}{2012a}]{CDW12}
Cai, Y., C.~Daskalakis, and S.~M. Weinberg (2012a).
\newblock An algorithmic characterization of multi-dimensional mechanisms.
\newblock In {\em Proceedings of the 44th Symposium on Theory of Computing
  Conference, {STOC}}, pp.\  459--478.

\bibitem[\protect\citeauthoryear{Cai, Daskalakis, and Weinberg}{Cai
  et~al.}{2012b}]{CDW12b}
Cai, Y., C.~Daskalakis, and S.~M. Weinberg (2012b).
\newblock Optimal multi-dimensional mechanism design: Reducing revenue to
  welfare maximization.
\newblock In {\em Proceedings of the 53rd Annual Symposium on Foundations of
  Computer Science (FOCS)}, pp.\  130--139.

\bibitem[\protect\citeauthoryear{Che, Kim, and Mierendorff}{Che
  et~al.}{2013}]{CKM13}
Che, Y.-K., J.~Kim, and K.~Mierendorff (2013).
\newblock Generalized reduced form auctions: A network flow approach.
\newblock {\em Econometrica\/}~{\em 81}, 2487--2520.

\bibitem[\protect\citeauthoryear{Chen, Diakonikolas, Paparas, Sun, and
  Yannakakis}{Chen et~al.}{2014}]{C+14}
Chen, X., I.~Diakonikolas, D.~Paparas, X.~Sun, and M.~Yannakakis (2014).
\newblock The complexity of optimal multidimensional pricing.
\newblock In {\em Proceedings of the Twenty-Fifth Annual {ACM-SIAM} Symposium
  on Discrete Algorithms}, pp.\  1319--1328.

\bibitem[\protect\citeauthoryear{Chow}{Chow}{1961}]{Chow61}
Chow, C. (1961).
\newblock On the characterization of threshold functions.
\newblock In {\em Proc. of the Symposium on Switching Circuit Theory and
  Logical Design (FOCS)}, pp.\  34--38.

\bibitem[\protect\citeauthoryear{Daskalakis, Deckelbaum, and Tzamos}{Daskalakis
  et~al.}{2014}]{DDT14}
Daskalakis, C., A.~Deckelbaum, and C.~Tzamos (2014).
\newblock The complexity of optimal mechanism design.
\newblock In {\em Proceedings of the Twenty-Fifth Annual {ACM-SIAM} Symposium
  on Discrete Algorithms}, pp.\  1302--1318.

\bibitem[\protect\citeauthoryear{De, Diakonikolas, Feldman, and Servedio}{De
  et~al.}{2014}]{DDFS}
De, A., I.~Diakonikolas, V.~Feldman, and R.~A. Servedio (2014).
\newblock Nearly optimal solutions for the chow parameters problem and
  low-weight approximation of halfspaces.
\newblock {\em J. {ACM}\/}~{\em 61\/}(2), 11.

\bibitem[\protect\citeauthoryear{De, Diakonikolas, and Servedio}{De
  et~al.}{2013}]{DDS}
De, A., I.~Diakonikolas, and R.~A. Servedio (2013).
\newblock A robust khintchine inequality, and algorithms for computing optimal
  constants in fourier analysis and high-dimensional geometry.
\newblock In {\em Automata, Languages, and Programming - 40th International
  Colloquium, {ICALP} 2013, Riga, Latvia, July 8-12, 2013, Proceedings, Part
  {I}}, pp.\  376--387.

\bibitem[\protect\citeauthoryear{Dobzinski, Fu, and Kleinberg}{Dobzinski
  et~al.}{2011}]{DFK11}
Dobzinski, S., H.~Fu, and R.~D. Kleinberg (2011).
\newblock Optimal auctions with correlated bidders are easy.
\newblock In L.~Fortnow and S.~P. Vadhan (Eds.), {\em Proceedings of the 43rd
  {ACM} Symposium on Theory of Computing, {STOC}}, pp.\  129--138.

\bibitem[\protect\citeauthoryear{Elkind}{Elkind}{2007}]{E07}
Elkind, E. (2007).
\newblock Designing and learning optimal finite support auctions.
\newblock In {\em Proceedings of the Eighteenth Annual {ACM-SIAM} Symposium on
  Discrete Algorithms, {SODA}}, pp.\  736--745.

\bibitem[\protect\citeauthoryear{Gershkov, Goeree, Kushnir, Moldovanu, and
  Shi}{Gershkov et~al.}{2013}]{G13}
Gershkov, A., J.~K. Goeree, A.~Kushnir, B.~Moldovanu, and X.~Shi (2013).
\newblock On the equivalence of bayesian and dominant strategy implementation.
\newblock {\em Econometrica\/}~{\em 81\/}(1), 197--220.

\bibitem[\protect\citeauthoryear{Gr{\"o}tschel, Lov{\'a}sz, and
  Schrijver}{Gr{\"o}tschel et~al.}{1988}]{GLS}
Gr{\"o}tschel, M., L.~Lov{\'a}sz, and A.~Schrijver (1988).
\newblock {\em Geometric Algorithms and Combinatorial Optimization}.
\newblock Springer.
\newblock Second Edition, 1993.

\bibitem[\protect\citeauthoryear{Hart and Reny}{Hart and Reny}{2014}]{HR14}
Hart, S. and P.~J. Reny (2014).
\newblock Implementation of reduced form mechanisms: A simple approach and a
  new characterization.
\newblock {\em Economic Theory Bulletin\/}.

\bibitem[\protect\citeauthoryear{Khachiyan}{Khachiyan}{1979}]{K79}
Khachiyan, L.~G. (1979).
\newblock A polynomial algorithm in linear programming.
\newblock {\em Soviet Mathematics Doklady\/}~{\em 20\/}(1), 191--194.

\bibitem[\protect\citeauthoryear{Maskin and Riley}{Maskin and
  Riley}{1984}]{MR84}
Maskin, E. and J.~Riley (1984).
\newblock Optimal auctions with risk-adverse buyers.
\newblock {\em Econometrica\/}~{\em 52}, 1473--1518.

\bibitem[\protect\citeauthoryear{Matthews}{Matthews}{1984}]{M84}
Matthews, S.~A. (1984).
\newblock On the implementability of reduced form auctions.
\newblock {\em Econometrica\/}~{\em 52}, 1519--1522.

\bibitem[\protect\citeauthoryear{Mierendorff}{Mierendorff}{2011}]{M11}
Mierendorff, K. (2011).
\newblock Asymmetric reduced form auctions.
\newblock {\em Economics Letters\/}~{\em 110}, 41--44.

\bibitem[\protect\citeauthoryear{Myerson}{Myerson}{1981}]{M81}
Myerson, R. (1981).
\newblock Optimal auction design.
\newblock {\em Mathematics of Operations Research\/}~{\em 6\/}(1), 58--73.

\bibitem[\protect\citeauthoryear{Nisan}{Nisan}{2007}]{N07}
Nisan, N. (2007).
\newblock Introduction to mechanism design (for computer scientists).
\newblock In N.~Nisan, T.~Roughgarden, E.~Tardos, and V.~Vazirani (Eds.), {\em
  Algorithmic Game Theory}. Cambridge University Press.

\bibitem[\protect\citeauthoryear{O'Donnell}{O'Donnell}{2014}]{O'Donnell}
O'Donnell, R. (2014).
\newblock {\em Analysis of Boolean Functions}.
\newblock Cambridge University Press.

\bibitem[\protect\citeauthoryear{O'Donnell and Servedio}{O'Donnell and
  Servedio}{2011}]{OS08}
O'Donnell, R. and R.~A. Servedio (2011).
\newblock The chow parameters problem.
\newblock {\em {SIAM} J. Comput.\/}~{\em 40\/}(1), 165--199.

\bibitem[\protect\citeauthoryear{Papadimitriou and Pierrakos}{Papadimitriou and
  Pierrakos}{2011}]{PP11}
Papadimitriou, C.~H. and G.~Pierrakos (2011).
\newblock On optimal single-item auctions.
\newblock In {\em Proceedings of the 43rd {ACM} Symposium on Theory of
  Computing}, pp.\  119--128.

\bibitem[\protect\citeauthoryear{Pitowsky}{Pitowsky}{1994}]{P94}
Pitowsky, I. (1994).
\newblock George boole's ‘conditions of possible experience’and the quantum
  puzzle.
\newblock {\em The British Journal for the Philosophy of Science\/}~{\em
  45\/}(1), 95--125.

\bibitem[\protect\citeauthoryear{Schrijver}{Schrijver}{1986}]{S86}
Schrijver, A. (1986).
\newblock {\em Theory of Linear and Integer Programming}.
\newblock Wiley.

\bibitem[\protect\citeauthoryear{Valiant}{Valiant}{1979}]{valiant}
Valiant, L.~G. (1979).
\newblock The complexity of enumeration and reliability problems.
\newblock {\em {SIAM} J. Comput.\/}~{\em 8\/}(3), 410--421.

\bibitem[\protect\citeauthoryear{Ziegler}{Ziegler}{1993}]{ziegler}
Ziegler, G.~M. (1993).
\newblock {\em Lectures on Polytopes}.
\newblock Springer.

\end{thebibliography}

\appendix

\section{Hardness of computing the Khintchine constant}
\label{a:khintchine}
 
\begin{prevproof}{Lemma}{l:wb}
$\Part$ is a well-known $\NP$-complete problem whose input consists of
$n$ integers $w_1,\ldots,w_n$ and the goal is to split the numbers in
two parts so that their sum is equal. This is equivalent to asking if
there exists $x \in \pmo^n$ such that $w\cdot x =0$ where $w
=(w_1,\ldots,w_n)$. The counting version which we denote $\sPart$ is
the problem of computing
\[ \Pr_{x\in \pmo^n}[w\cdot x =0 ]. \]
It is complete for $\sharpP$. We show that this problem reduces to computing the Khintchine constant.

Given $w$, define the vectors  
\begin{align*} 
a_0 = (2w_1,\ldots,2w_n,0),  \ \ a_1 = (2w_1,\ldots,2w_n, 1).
\end{align*}
For $x \in \pmo^n$, define $x^+ =(x,1)$ and $x^- = (x,-1)$. We observe that
\begin{align}
\label{eq:a0}|a_0\cdot x^+| + |a_0\cdot x^-| &  = 4|w\cdot x|;\\
\label{eq:a_1}|a_1\cdot x^+| + |a_1\cdot x^-| & = |2w\cdot x +1| + |2w\cdot x -1|.
\end{align}
When $w\cdot x =0$, the RHS equals of Equation \eqref{eq:a0} equals $0$, while that of Equation \eqref{eq:a_1} equals $2$.

If $w\cdot x \neq 0$, then $|2w\cdot x| \geq 2$. We then use 
\[ |a + b| + |a - b| = 2\max(|a|,|b|)\]
to conclude that 
\begin{align*}
|a_1\cdot x^+| + |a_1\cdot x^-|  = 4|w\cdot x| = |a_0\cdot x^+| + |a_0\cdot x^-|.
\end{align*}
Finally, observe that we can write
\begin{align*}
K(a_i) = \expect[\xu]{\frac{|a_i\cdot x^+| + |a_i \cdot x^-|}{2}}.
\end{align*}
It follows that
\[ \Pr_{x\in \pmo^n}[w\cdot x =0 ] = K(a_1) - K(a_0).\]
\end{prevproof}

\section{Completing the Proof of Lemma ~\ref{lem:chow}}
\label{app:exp}

Consider the affine function $a(x) = a_0 +
\sum_{i=1}^na_i(-1)^{1+x_i}$ mapping $\zo^n$ to $\R$.  
Our goal is to show that $\expect[\xo]{\sign(a(x))a(x)} = \frac{1}{2}(K(a) + a_0)$.
Observe that
\begin{align*}
(2\sign(a(x)) -1)a(x) = |a(x)|
\end{align*}
Hence 
\begin{align*}
\E_{x \in \zo^n}[|a(x)|] & = 2\E_{x \in \zo^n}[\sign(a(x))a(x)] - \E_{x \in \zo^n}[a(x)] \\
& = 2\E_{x \in \zo^n}[\sign(a(x))a(x)] - a_0
\end{align*}
We will now show that
\[ \E_{x \in \zo^n}[|a(x)|]  = K(a),\]
which implies the claim. Observe that
\begin{align*}
K(a) & = \E_{y \in \pmo^{n+1}}[|a\cdot y|]\\
& = \frac{1}{2}\E_{y \in \pmo^n}[|a_0 + \sum_{i=1}^na_iy_i|] + \frac{1}{2}\E_{x
  \in \pmo^n}[|-a_0 + \sum_{i=1}^na_iy_i|]
\end{align*}
We claim that the two expectations on the RHS are in fact equal to
each other, since
\begin{align*}
\E_{y \in \pmo^n}[|-a_0 + \sum_{i=1}^na_iy_i|] = \E_{y
  \in \pmo^n}[|a_0 -\sum_{i=1}^na_iy_i|] = \E_{y  \in \pmo^n}[|a_0 + \sum_{i=1}^na_iy_i|]
\end{align*}
since $\sum_ia_iy_i$ is an even random variable. Hence we get
\begin{align*}
K(a) & = \E_{y \in \pmo^n}\left[|a_0 + \sum_{i=1}^na_iy_i|\right]\\ 
& = \E_{x \in \zo^n}\left[|a_0 + \sum_{i=1}^na_i(-1)^{1 + x_i}|\right]\\
& = \E_{x \in\zo^n}[|a(x)|]. 
\end{align*}

\end{document}